\documentclass[a4paper]{article}
\usepackage{hyperref}
\usepackage{amsthm}
\usepackage{authblk}
\usepackage{tikz}
\usepackage{amsmath}
\usepackage{graphicx}
\usepackage[ruled,linesnumbered,algo2e,vlined]{algorithm2e}
\usepackage{cleveref}
\usepackage[whole]{bxcjkjatype}
\usepackage[normalem]{ulem}

\usepackage{colortbl}

\newcommand\upuline{\bgroup\markoverwith{\hbox{\kern-0.1em \vbox{\hrule width.3em \kern-0.4ex}}}\ULon}
\newcommand\downuline{\bgroup\markoverwith{\hbox{\kern-0.1em \vbox{\hrule width.3em \kern-0.7ex}}}\ULon}
\newcommand\doubleuline{\bgroup\markoverwith{\hbox{\kern-0.1em \vbox{\hrule width.3em \kern-0.4ex \hrule \kern-0.4ex }}}\ULon}

\newcommand{\etal}{et al.}

\newcommand{\LCover}{\mathit{LCover}}
\newcommand{\SCover}{\mathit{SCover}}
\newcommand{\Reach}{\mathit{Reach}}
\newcommand{\Border}{\mathit{Border}}

\newcommand{\Cov}{\mathsf{Cov}_{\approx}}
\newcommand{\Bord}{\mathsf{Bord}_{\approx}}
\newcommand{\LSeed}{\mathsf{LSeed}_{\approx}}

\newcommand{\Dead}{\mathit{Dead}}
\newcommand{\LiveChildren}{\mathit{LSChildren}}
\newcommand{\LargestLive}{\mathit{LongestLSAnc}}
\newcommand{\LLive}{\mathsf{LongestLSeedCov}}

\newcommand{\True}{\mathbf{True}}

\newcommand{\ABorder}{$\approx$-border}
\newcommand{\ABorders}{$\approx$-borders}
\newcommand{\AOccur}{$\approx$-occurrence}

\newcommand{\ACover}{$\approx$-cover}
\newcommand{\ACovers}{$\approx$-covers}
\newcommand{\ALeftSeed}{left $\approx$-seed}
\newcommand{\ALeftSeeds}{left $\approx$-seeds}
\newcommand{\OCC}{\mathsf{Occ}}
\newcommand{\LSC}{\msf{LSChildren}}
\newcommand{\lsc}{\msf{rangeChildren}}
\newcommand{\mtt}[1]{\mathtt{#1}}
\newcommand{\msf}[1]{\mathsf{#1}}

\newcommand{\OPapprox}{\stackrel{\mathrm{op}}{\approx}}
\newcommand{\PRapprox}{\stackrel{\mathrm{pr}}{\approx}}

\SetKwProg{Fn}{Function}{}{end}

\def\tablename{Table}
\newtheorem{definition}{Definition}
\newtheorem{lemma}{Lemma}
\newtheorem{theorem}{Theorem}
\newtheorem{corollary}{Corollary}

\makeatletter
\renewenvironment{proof}[1][\proofname]{\par
  \normalfont
  \topsep6\p@\@plus6\p@ \trivlist
  \item[\hskip\labelsep{\bfseries #1}\@addpunct{\bfseries.}]\ignorespaces
}{%
  \endtrivlist
}
\renewcommand{\proofname}{Proof}
\makeatother

\theoremstyle{definition}
\newtheorem{example}{Example}

\title{Computing Covers under Substring Consistent Equivalence Relations} 

\author[]{Natsumi~Kikuchi}
\author[]{Diptarama~Hendrian}
\author[]{Ryo~Yoshinaka}
\author[]{Ayumi~Shinohara}
\affil[]{Graduate school of information sciences, Tohoku University, Japan}
\affil[]{\small{\texttt{natsumi\_kikuchi@shino.ecei.tohoku.ac.jp}\\
    \texttt{\{diptarama, ryoshinaka, ayumis\}@tohoku.ac.jp}}}
\date{}

\begin{document}

\maketitle

\begin{abstract}
	Covers are a kind of quasiperiodicity in strings.
A string $C$ is a cover of another string $T$ if any position of $T$ is inside some occurrence of $C$ in $T$.
The shortest and longest cover arrays of $T$ have the lengths of the shortest and longest covers of each prefix of $T$, respectively.
The literature has proposed linear-time algorithms computing longest and shortest cover arrays taking border arrays as input.
An equivalence relation $\approx$ over strings is called a \emph{substring consistent equivalence relation (SCER)}
iff $X \approx Y$ implies (1) $|X| = |Y|$ and (2) $X[i:j] \approx Y[i:j]$ for all $1 \le i \le j \le |X|$.
In this paper, we generalize the notion of covers for SCERs
and prove that existing algorithms to compute the shortest cover array and the longest cover array of a string $T$ under the identity relation will work for any SCERs taking the accordingly generalized border arrays. 
\end{abstract}

\section{Introduction}
Finding regularities in strings is an important task in string processing due to its applications such as pattern matching and string compression.
Many variants of regularities in strings have been studied
including periods, covers, and seeds \cite{Apostolico1993,Apostolico1991,Iliopoulos1996}.
One of the most studied regularities is periods due to their
mathematical combinatoric properties and their applications to string processing algorithms~\cite{Crochemore2002}.
The notion of periods has been generalized concerning various kinds of equivalence relations.
Apostolico and Giancarlo~\cite{Apostolico2008} studied periods on parameterized strings.
Gourdel et al.~\cite{Gourdel2020} studied string periods on the order-preserving model.

Covers are another kind of regularities that have extensively been studied.
For two strings $T$ and $C$,
$C$ is a \emph{cover} of $T$ if any position of $T$ is inside some occurrences of $C$ in $T$.
For example, $\mathtt{aba}$ is a cover of $T = \mathtt{\upuline{aba}\downuline{ab}\doubleuline{a}\upuline{ba}\downuline{ab}\doubleuline{a}\upuline{ba}\downuline{aba}}$
because all positions in $T$ are inside occurrences of $\mathtt{aba}$.
The other covers of $T$ are
$\mathtt{abaaba}$, $\mathtt{abaababaaba}$ and $T$ itself.
Apostolico and Ehrenfeucht~\cite{Apostolico1993}
called a string having a cover besides itself \emph{quasiperiodic}
and proposed an algorithm that computes all maximal quasiperiodic substrings of a string.
Later, Iliopoulos and Mouchard~\cite{Iliopoulos1999} and Brodal and Pedersen~\cite{Brodal2000} proposed $O(n\log n)$ time algorithm for this task. 
Apostolico \etal{}~\cite{Apostolico1991} presented a linear-time algorithm to test whether a string 
is quasiperiodic.
Breslauer~\cite{Breslauer1992} proposed an online linear-time algorithm that computes the shortest covers of all prefixes as \emph{the shortest cover array} of a string.
Moore and Smyth~\cite{Moore1994,Moore1995} proposed a linear-time algorithm to compute all covers of a string.
Later, Li and Smyth~\cite{Li2002} proposed an online linear-time algorithm to compute the longest proper covers of all prefixes of a string as \emph{the longest cover array}.
Amir et al.~\cite{Amir2019approx} defined the approximate cover problem and showed its NP-hardness.

Recently, Matsuoka~\etal{}~\cite{Matsuoka2016} introduced the notion of \emph{substring consistent equivalence relations} (\emph{SCERs}),
 which are equivalence relations $\approx$ on strings such that $X \approx Y$ implies (1) $|X| = |Y|$ and (2) $X[i:j] \approx Y[i:j]$ for all $1 \le i \le j \le |X|$, where $X[i:j]$ denotes the substring of $X$ starting at $i$ and ending at $j$.
Clearly the identity relation is an SCER.
Moreover, many variants of equivalence relations used in pattern matching are SCERs, such as parameterized pattern matching~\cite{Baker1996}, order-preserving pattern matching~\cite{Kim2014,Kubica2013}, permuted pattern matching~\cite{Katsura2013},
and Cartesian tree matching~\cite{Park2019}.
Matsuoka~\etal{}~\cite{Matsuoka2016} proposed an algorithm to compute the border array of an input string $T$ under an SCER, which can be used for pattern matching under SCERs.

In this paper, we generalize the notion of covers, which used to be defined based on the identity relation, to be based on SCERs, and prove that both of the algorithms for the shortest and longest cover arrays by Breslauer~\cite{Breslauer1992} and Li and Smyth~\cite{Li2002}, respectively, work under SCERs with no changes:
just by replacing the input of those algorithms from the border array under the identity relation to the one under a concerned SCER, their algorithms compute the shortest and longest cover arrays under the SCER.
As a minor contribution, we present a slightly simplified version of Li and Smyth's algorithm, with a correctness proof.
\Cref{table:complexity_summary} summarizes implications of our results.
The time complexities for computing shortest and longest cover arrays based on various SCERs are the same as those for border arrays.
Moreover, if border arrays under an equivalence relation
can be computed online, e.g., parameterized equivalence and order-isomorphism,
these cover arrays can be computed online by computing border arrays with existing online algorithms at the same time.


\begin{table}[t]
	\caption{The time complexity of computing border ($\Border$), shortest cover ($\SCover$) and longest cover ($\LCover$) arrays under SCERs, where $n$ is the input length, $\Pi$ is the parameter set in parameterized equivalence, and $k$ is the number of input strings in permuted equivalence.
    }
	\label{table:complexity_summary}
	\centering
	\setlength{\tabcolsep}{3pt}
	\begin{tabular}{|l|c|c|c|c|}
		\hline
		Equivalence relation & $\Border$ & $\SCover$ & $\LCover$ \\
		\hline
		Identity equivalence & $O(n)$~\cite{Knuth1977} & $O(n)$~\cite{Breslauer1992} & $O(n)$~\cite{Li2002} \\
		Parameterized equivalence & $O(n \log |\Pi|)$~\cite{Amir1994} & \cellcolor[gray]{0.8}$O(n \log |\Pi|)$ & \cellcolor[gray]{0.8}$O(n \log |\Pi|)$ \\
		Order-isomorphism & $O(n \log{n})$~\cite{Kim2014,Kubica2013} & \cellcolor[gray]{0.8}$O(n \log{n})$ & \cellcolor[gray]{0.8}$O(n \log{n})$ \\
		Permuted equivalence & $O(nk)$~\cite{Diptarama2016,Hendrian2019} & \cellcolor[gray]{0.8}$O(nk)$ & \cellcolor[gray]{0.8}$O(nk)$ \\
		\hline
	\end{tabular}
\end{table}


\section{Preliminaries}
For an alphabet $\Sigma$, $\Sigma^*$ denotes the set of all strings over $\Sigma$, including the empty string $\varepsilon$.
The length of a string $T \in \Sigma^*$ is denoted as $|T|$.
For $1 \le i \le j \le |T|$, $T[i:j]$ denotes the substring of $T$ that starts at $i$ and ends at $j$.
By $T[:j] = T[1:j]$ we denote the \emph{prefix} of $T$ that ends at $j$ and
by $T[i:] = T[i:|T|]$ the \emph{suffix} of $T$ that starts at $i$.

Matsuoka et al.~\cite{Matsuoka2016} introduced the notion of substring consistent equivalence relations, generalizing several equivalence relations proposed so far in pattern matching.
\begin{definition}[Substring Consistent Equivalence Relation (SCER) $\approx$] \label{def:scer}
	An equivalence relation ${\approx} \subseteq \Sigma^* \times  \Sigma^*$ is an \emph{SCER} if
	for two strings $X$ and $Y$, $X \approx Y$ implies (1) $|X| = |Y|$ and (2) $X[i:j] \approx Y[i:j]$ for all $1 \le i \le j \le |X|$.
	By $[X]_\approx$ we denote the $\approx$-equivalence class of $X$.
\end{definition}
For instance, matching relations in parameterized pattern matching~\cite{Baker1996}, order-preserving pattern matching~\cite{Kim2014,Kubica2013},
and permuted pattern matching~\cite{Katsura2013} are SCERs,
while matching relations in abelian pattern matching~\cite{Ehlers2015k}, indeterminate string pattern matching~\cite{Antoniou2008} and function matching~\cite{Amir2006} are not.

\begin{definition}[Parameterized equivalence~\cite{Baker1996}]\label{def:pr}
    Two strings $X$ and $Y$ of the same length
    are a \emph{parameterized match}, denoted as $X \PRapprox Y$, if $X$ can be transformed into $Y$ by applying a renaming bijection $g$ from the characters of\/ $X$ to the characters of\/ $Y$.
\end{definition}
\begin{definition}[Order-isomorphism~\cite{Kim2014,Kubica2013}]\label{def:oi}
	Two strings $X$ and $Y$ of the same length over an alphabet with a linear order $\prec$ are \emph{order isomorphic}, denoted as $X \OPapprox Y$,
	if $X[i] \prec X[j] \Leftrightarrow Y[i] \prec Y[j]$ for all $1 \le i , j \le |X|$. 
\end{definition}
\begin{definition}[\AOccur{}~\cite{Matsuoka2016}]
	For two strings $T$ and $P$, a position $1 \le i \le |T| - |P| +1$ is an \emph{\AOccur{}} of $P$ in $T$ if $P \approx T[i:i+|P|-1]$.
	The set of \AOccur{} positions of $P$ in $T$ is denoted by $\OCC_{P,T}$.
\end{definition}

\begin{definition}[\ABorder{}~\cite{Matsuoka2016}]\label{def:border}
	A string $B$ is a \emph{$\approx$-border} of $T$ if\/ $B \approx T[:|B|] \approx T[|T|-|B|+1:]$.
	We denote by $\Bord(T)$ the set of all \ABorders{} of $T$.
	A \ABorder{} $B$ of\/ $T$ is called \emph{proper} if\/ $|B| < |T|$, and called \emph{trivial} if\/ $B = \varepsilon$.
\end{definition}

\begin{lemma}[\cite{Matsuoka2016}]\label{lem:bordertext}
	(1) $B \in \Bord(S)$ and $B' \in \Bord(B)$ implies $B' \in \Bord(S)$.
	(2) $B, B' \in \Bord(S)$ and $|B'| \le |B|$ implies $B' \in \Bord(B)$.
\end{lemma}
Based on Lemma~\ref{lem:bordertext}, Matsuoka et al.~\cite{Matsuoka2016} proposed an algorithm to compute border arrays under SCERs, which are defined as follows.
\begin{definition}[\ABorder{} array]
	The \emph{\ABorder{} array} $\Border_T$ of\/ $T$ is an array of length $|T|$ such that
	$\Border_T[i]=\max\{\,|B| \mid B \textup{ is a proper \ABorder{} of } T[:i]\,\}$ for $1 \le i \le |T|$.
\end{definition}
Tables \ref{table:example} and \ref{table:example_parameter} show examples of \ABorder{} arrays.
We use the identity relation in \Cref{table:example} and the parameterized equivalence (\Cref{def:pr}) in \Cref{table:example_parameter}.

The well-known property on $=$-borders (e.g., \cite{Aho1974design}) holds for \ABorder{s}, too.
\begin{lemma}\label{lem:barrayproperty}
	For any $1 < i \le n$,
	$\Border_{T}[i-1]+1 \ge \Border_{T}[i]$.
\end{lemma}
\newcommand{\PROOFbarrayproperty}{
\begin{proof}
    Let $b = \Border_{T}[i]$.
    If $b=0$, the lemma holds by $\Border_{T}[i-1] \ge 0$.
    Otherwise,
    since $T[:b] \approx T[i-b+1:i]$ by \Cref{def:border},
    we have $T[:b-1] \approx T[i-b+1:i-1]$ by \Cref{def:scer}.
    Thus $T[:b-1] \in \Bord(T[:i-1])$ by \Cref{def:border}.
	By $b < i$, $T[:b-1]$ is a proper \ABorder{} of $T[:i-1]$, and thus $\Border_{T}[i-1] \ge b-1 $.
    Therefore, $\Border_{T}[i-1]+1 \ge b = \Border_{T}[i]$.
\qed\end{proof}
}
%
%

\begin{table}[t]
	\caption{The $=$-border array, the shortest $=$-cover array, and the longest $=$-cover array of $T=\mathtt{abaababaabaababa}$.}
	\label{table:example}
	\centering
	\setlength{\tabcolsep}{5pt}
	\begin{tabular}{|c|c|c|c|c|c|c|c|c|c|c|c|c|c|c|c|c|}
		\hline
		 & 1 & 2 & 3 & 4 & 5 & 6 & 7 & 8 & 9 & 10 & 11 & 12 & 13 & 14 & 15 & 16 \\
		\hline
		$T$ & \texttt{a} & \texttt{b} & \texttt{a} & \texttt{a} & \texttt{b} & \texttt{a} & \texttt{b} & \texttt{a} & \texttt{a} & \texttt{b} & \texttt{a} & \texttt{a} & \texttt{b} & \texttt{a} & \texttt{b} & \texttt{a}\\
		\hline
		$\Border_T$ & 0 & 0 & 1 & 1 & 2 & 3 & 2 & 3 & 4 & 5 & 6 & 4 & 5 & 6 & 7 & 8\\
		$\SCover_T$ & 1 & 2 & 3 & 4 & 5 & 3 & 7 & 3 & 9 & 5 & 3 & 12 & 5 & 3 & 15 & 3\\
		$\LCover_T$ & 0 & 0 & 0 & 0 & 0 & 3 & 0 & 3 & 0 & 5 & 6 & 0 & 5 & 6 & 0 & 8\\
		\hline
	\end{tabular}
\end{table}

\begin{table}[t]
	\caption{The $\PRapprox$-border array, the shortest $\PRapprox$-cover array, and the longest $\PRapprox$-cover array of $T=\mathtt{abaababaabaababa}$. Notice that $\SCover_T[i] = 1$ for all $i$, for $\texttt{a} \PRapprox \texttt{b}$.}
	\label{table:example_parameter}
	\centering
	\setlength{\tabcolsep}{5pt}
	\begin{tabular}{|c|c|c|c|c|c|c|c|c|c|c|c|c|c|c|c|c|}
		\hline
		& 1 & 2 & 3 & 4 & 5 & 6 & 7 & 8 & 9 & 10 & 11 & 12 & 13 & 14 & 15 & 16 \\
		\hline
		$T$ & \texttt{a} & \texttt{b} & \texttt{a} & \texttt{a} & \texttt{b} & \texttt{a} & \texttt{b} & \texttt{a} & \texttt{a} & \texttt{b} & \texttt{a} & \texttt{a} & \texttt{b} & \texttt{a} & \texttt{b} & \texttt{a}\\
		\hline
		$\Border_T$ & 0 & 1 & 2 & 1 & 2 & 3 & 3 & 3 & 4 & 5 & 6 & 4 & 5 & 6 & 7 & 8\\
		$\SCover_T$ & 1 & 1 & 1 & 1 & 1 & 1 & 1 & 1 & 1 & 1 & 1 & 1 & 1 & 1 & 1 & 1\\
		$\LCover_T$ & 0 & 1 & 2 & 1 & 2 & 3 & 3 & 3 & 1 & 5 & 6 & 1 & 5 & 6 & 3 & 8\\
		\hline
	\end{tabular}
\end{table}

\section{Covers under SCERs}
In this section, we define covers under SCERs (\ACovers) and present some properties of \ACovers, which prepares for the succeeding sections.
Section~\ref{sec:shortest} shows that the algorithm to compute shortest cover arrays by Breslauer~\cite{Breslauer1992} will work under SCERs with no change.
Section~\ref{sec:longest} presents a slight variant of the algorithm by Li and Smyth~\cite{Li2002} for computing the longest cover arrays and proves its correctness.
\begin{definition}[\ACover{}] \label{def:cover}
	We say that a string $C$ of length $c$ is an \emph{\ACover{}} of a string $T$ of length $n$ if
	there are $x_1,x_2,\dots,x_m \in \OCC_{C,T}$ such that $x_1 = 1$, $x_m = n-c+1$ and $x_{i-1} < x_i \le x_{i-1} + c$ for all $1 < i \le m$.
	Moreover, we say that an \ACover{} $C$ of\/ $T$ is \emph{proper} if\/ $c < n$.
	The set of all \ACovers{} of\/ $T$ is denoted by $\Cov(T)$.
	A string $T$ is \emph{primitive}\footnote{In some references it is called \emph{superprimitive},
    reserving the term ``primitive'' for strings that cannot be represented as $S^k$ for some string $S$ and integer $k \ge 2$.}
    if\/ $T$ has no proper \ACover{}.
\end{definition}
By definition, $\Cov(T) \subseteq \Bord(T)$.
Below we observe that basic lemmas in~\cite{Breslauer1992} on $=$-covers and $=$-borders hold for \ACovers{} and \ABorders{}.
\begin{lemma}
	\label{lem:coverisbordercover}
	If\/ $C \in \Cov(T)$, $B \in \Bord(T)$, and $|C| \le |B|$, then $C \in \Cov(B)$.
\end{lemma}
\newcommand{\PROOFcoverisbordercover}{
\begin{proof}
	By $\OCC_{C,B} = \{\, i \in \OCC_{C,T} \mid i+m-1 \le b \,\}$, it suffices to show $b-m+1 \in \OCC_{C,B}$.
	This follows the fact that $C \approx T[n-m+1:] \approx B[b-m+1:]$.
\qed\end{proof}
}
%
\begin{lemma}\label{lem:shortcoverlong}
	For any $C,C'\in \Cov(T)$ such that $|C| \le |C'|$,
	$C \in \Cov(C')$.
\end{lemma}
\newcommand{\PROOFshortcoverlong}{
\begin{proof}
By Lemma~\ref{lem:coverisbordercover} and $C' \in \Cov(T) \subseteq \Bord(T)$.
\qed\end{proof}
}
\begin{lemma}\label{lem:covercover}
	If $C \in \Cov(T)$ and $C' \in \Cov(C)$, then $C' \in \Cov(T)$.
\end{lemma}
\newcommand{\PROOFcovercover}{
\begin{proof}
    Let $c = |C|$ and $c' = |C'|$.
	By \Cref{def:cover}, $C \in \Bord(T)$ and $C' \in \Bord(C)$.
	Thus, $C' \in \Bord(T)$ by \Cref{lem:bordertext} (2).
	From \Cref{def:border}, we have $1, n-c'+1 \in \OCC_{C',T}$.

    For $z \in \OCC_{C',T} \setminus \{1\}$, let $x = \max\{\, x \in\OCC_{C,T} \mid x < z\,\}$.
    Then either (1) $x < z < z + c' \leq x + c$ or (2) $x < z \leq x + c < z + c'$ holds.
    In the case (1), since $C' \in \Cov(C)$, there exists $z' \in \OCC_{C',T}$ such that $z' < z \leq z' + c'$.
    In the case (2), since $C' \in \Cov(C)$, we have $C' \approx C[c-c'+1:] \approx T[x+c-c':x+c-1]$.
    For $z' = x+c-c'$,
    we get $z' \in \OCC_{C',T}$ and $z' < z \leq z' + c'$.
	Therefore, we get $C' \in \Cov(T)$.
\qed\end{proof}
}
\begin{lemma}\label{lem:primitivecover}
An \ACover{} $C$ of $T$ is primitive iff it is a shortest \ACover{} of $T$.
\end{lemma}
\begin{lemma}\label{lem:cover_overlap}
For $0 \le i-1 \le j \le |T|$,  $\Cov(T[:j]) \cap \Cov(T[i:]) \subseteq \Cov(T)$.
\end{lemma}

\begin{lemma}
	\label{lem:forreach}
	A string $C$ of length $c$ is a proper \ACover{} of\/ $T$ of length $n$ iff\/ $C \in \Bord(T)$ and $C \in \Cov(T[:n-i])$ for some $1 \le i \le c$.
\end{lemma}
\newcommand{\PROOFforreach}{
\begin{proof}
	$(\Longrightarrow)$
	Since $C$ is a proper \ACover{} of $T$,
	$\OCC_{C,T} \setminus \{1\} \neq \emptyset$ and
	$C \in \Bord(T)$ by \Cref{def:cover}.
    Let $x_k = \max(\OCC_{C,T}) = n-m+1 > 1$ and
     $x_{k-1} = \max(\OCC_{C,T} \setminus \{x_k\})$.
	Then we have $x_{k-1} < x_{k} \leq x_{k-1} + m$ by \Cref{def:cover}.
	Clearly $C \in \Cov(T[:x_{k-1}+m-1])$ and $n - 1 \ge x_{k-1}+m-1 \ge n-m$.

	$(\Longleftarrow)$
	Assume that $C \in \Bord(T)$ and $C \in \Cov(T')$, where $T' = T[:n-i]$ for some $1 \le i \le m$.
	Since $C \in \Bord(T)$, we have $n-m+1 \in \OCC_{C,T}$.
	Let $\{x_1,\dots,x_{k'},\dots,x_k\}=\OCC_{C,T}$ and
	 $\{x_1,\dots,x_{k'}\}=\OCC_{C,T'}$ where $x_{1}<\dots<x_{k'}<\dots<x_k$, $x_1=1$, $x_{k'}=n-i-m+1$ and $x_k=n-m+1$.
	For all $1 < j \le k'$, we have $x_{j-1} < x_j \le x_{j-1}+m$ by $C \in \Cov(T')$.
	For all $j > k'$, we have $x_{j-1} < x_j \le x_k \le x_{k'}+m \le x_{j-1}+m$.
	Therefore, $C \in \Cov(T)$.
\qed\end{proof}
}


In the seaquel of this paper, we fix an input string $T$ of length $n$.

\section{Shortest \ACover{} array}\label{sec:shortest}

\begin{algorithm2e}[t]
	\caption{Algorithm computing the shortest \ACover{} array}
	\label{alg:shortest_cover}
	\SetVlineSkip{0.5mm}
	let $\Border$ be the \ABorder{} array of $T$\;
	$\Reach[i] \leftarrow 0$ for $1 \le i \le n$\; 
	\For{$1 \le i \le n$}{
		\If{\scalebox{0.97}[1]{$\Border[i] > 0$ $\mathbf{and}$ $\Reach[\SCover[\Border[i]]] \ge i - \SCover[\Border[i]]$}}{\label{alg:updatecondition}
			$\SCover[i] \leftarrow \SCover[\Border[i]]$\;
			$\Reach[\SCover[i]] \leftarrow i$\;\label{alg:updatereach}
		}
		\Else{
			$\SCover[i] \leftarrow i$\;
			$\Reach[i] \leftarrow i$\;
		}
	}
\end{algorithm2e}

In this section we prove that \Cref{alg:shortest_cover} by Breslauer~\cite{Breslauer1992} computes the shortest \ACover{} array for an input string $T$ based on the $\approx$-border array.

\begin{definition}[Shortest \ACover{} array]
	The \emph{shortest \ACover{} array} $\SCover_T$ of\/ $T$ is an array of length $n$ such that
	$\SCover_T[i]=\min\{\,|C| \mid C \in \Cov(T[:i])\,\}$ for $1 \le i \le n$.
\end{definition}
Tables \ref{table:example} and \ref{table:example_parameter} show examples of shortest \ACover{} arrays.
Note that $\SCover_T[i]$ is the length of the unique (modulo $\approx$-equivalence) primitive cover of $T[:i]$ by Lemma~\ref{lem:primitivecover}.


\Cref{alg:shortest_cover} uses an additional array $\Reach$ to compute $\SCover$.
The algorithm updates $\Reach$ and $\SCover$ incrementally so that $\Reach[j]$ shall be the length of the longest prefix of $T$ of which $T[:j]$ is a $\approx$-cover
and $\SCover$ shall be the shortest $\approx$-cover array.
More precisely,
in each iteration $i$, the algorithm updates $\Reach$ and $\SCover$ so that they satisfy the following properties at the end of the $i$-th iteration.
\begin{description}
	\item[] \textbf{R}($i$)
	$\Reach[j]=0$ if $j > i$ or $T[:j]$ is not primitive.
	Otherwise, $\Reach[j] = \max \{\,p \mid T[:j] \textup{ is a \ACover{} of } T[:p] \textup{ and } p \le i\,\} $.
    \item[]\textbf{S}{($i$)}
      For $1 \le j \le i$, $\SCover[j]=\min\{\,|C| \mid C \in \Cov(T[:j])\,\}$.
\end{description}
If \textbf{S}{($n$)} holds, we have $\SCover = \SCover_T$.
\begin{theorem}
	Given the \ABorder{} array of text\/ $T$ of length $n$,
	\Cref{alg:shortest_cover} computes the shortest \ACover{} array $\SCover_T$ of\/ $T$ in $O(n)$ time.
\end{theorem}

\begin{proof}
	The linear time complexity is obvious.

	We show the above invariants \textbf{R}{($i$)} and \textbf{S}{($i$)} by induction on $i$.
	Clearly the invariant holds for $i=0$, i.e., the initial values of $\Reach[j]=0$ for all $j >0$ satisfy the invariant \textbf{R}($0$).
	Vacuously \textbf{S}($0$) is true.

	Assume that \textbf{R}{($i-1$)} and \textbf{S}{($i-1$)} hold at the beginning of the $i$-th iteration.
    Let $b = \Border[i]$ and $c = \SCover[b]$.

	Suppose the \textbf{if}-condition of \Cref{alg:updatecondition} is satisfied in the $i$-th iteration.
	By the induction hypothesis on $\Reach[c]$, which is at least as large as $i-c \ge 1$ at the beginning of the $i$-th iteration,
	$T[:c]$ is a primitive \ACover{} of $T[:i-l]$ for some $1 \leq l \leq c$.
	Then the algorithm updates the value of $\Reach[c]$ to $i \ge 1$, which is still positive at the end of the $i$-th iteration.
	By $T[:b] \in \Bord(T[:i])$ and $T[:c] \in \Cov(T[:b]) \subseteq \Bord(T[:b])$ (by \textbf{S}($i-1$)),
	\Cref{lem:bordertext} (1) implies $T[:c] \in \Bord(T[:i])$.
    Therefore, $T[:c]$ is a proper \ACover{} of $T[:i]$ by \Cref{lem:forreach}.
    Thus, $\Reach[c] = i$ satisfies the  invariant.
	On the other hand, the value $\Reach[i]$ is not changed from its initial value $0$, while we get $\SCover[i]=c$.
	Indeed $T[:i]$ is not primitive as it has a \ACover{} $T[:c]$.
	That is, $\Reach[i]$ and $\SCover[i]$ satisfy the invariants.
	Since $T[:c]$ is the unique primitive \ACover{} prefix of $T[:i]$, for other $j$, $\Reach[j]$ need not be updated.

	Suppose the \textbf{if}-condition is not satisfied in the $i$-th iteration, where both $\Reach[i]$ and $\SCover[i]$ are set to be $i$.
	If $b = 0$,
	$T[:i]$ has no proper \ACover{}.
	Thus $T[:i]$ is primitive and the lemma holds.
    Next, consider the case where $b \neq 0$ and $\Reach[c] < i - c$.
    To show by contradiction that $T[:i]$ is primitive, assume that $T[:i]$ has a primitive proper \ACover{} $T[:k]$.
    By $T[:k] \in \Cov(T[:i]) \subseteq \Bord(T[:i])$ and Lemma~\ref{lem:coverisbordercover}, we have $T[:k] \in \Cov(T[:b])$.
    Since $T[:b]$ has only one (up to $\approx$-equivalence) primitive \ACover{} by Lemma~\ref{lem:primitivecover}, we have $k = c$, i.e., $T[:c] \in \Cov(T[:i])$.
    By Lemma~\ref{lem:forreach}, $T[:c] \in \Cov(T[:i-j])$ for some $1 \le j \le c$,
     which contradicts the fact $\Reach[c] < i - c$ with the induction hypothesis.
    Therefore, $T[:i]$ has no primitive proper \ACover{} and thus $T[:i]$ is primitive by Lemma~\ref{lem:primitivecover}.
    We conclude that $\Reach[i]=\SCover[i]=i$ satisfies $\textbf{R}(i)$ and $\textbf{S}(i)$
    and $\Reach[j]$ need not be updated for other $j$.
\qed\end{proof}

\begin{corollary}
	If $\Border_T$ can be computed in $\beta(n)$ time,
	$\SCover_T$ can be computed in $O(\beta(n) + n)$ time.
\end{corollary}

\section{Longest \ACover{} array}\label{sec:longest}
This section discusses computing the longest \ACover{} array of a text.
Tables~\ref{table:example} and~\ref{table:example_parameter} show examples of longest \ACover{} arrays.
\begin{definition}[Longest \ACover{} array]\label{def:longestcoverarray}
	The \emph{longest \ACover{} array} $\LCover_T$ of\/ $T$ is an array of length $n$ such that for $1 \le i \le n$,
$\LCover_T[i]=\max(\{\,|C| \mid C \textup{ is a proper \ACover{} of } T[:i]\,\}\cup \{0\})$.
\end{definition}
Let $\LCover^0_T[i] = i$ and $\LCover^q_T[i] = \LCover_T[\LCover^{q-1}_T[i]]$ for $q \ge 1$.
The following lemma is a corollary to Lemmas~\ref{lem:shortcoverlong} and~\ref{lem:covercover}.
\begin{lemma}
	\label{lem:longestrecursive}
	For any $1 \le j \le i$,
	 $T[:j] \in \Cov(T[:i])$ iff $j = \LCover^q_T[i]$ for some $q \ge 0$.
\end{lemma}
Therefore, using the longest \ACover{} array, one can easily obtain all the $\approx$-covers up to $\approx$-equivalence.

Li and Smyth~\cite{Li2002} presented an online linear-time algorithm to compute the longest $=$-cover array from the $=$-border array of a text $T$.
We will present a slight variant of theirs for computing the longest \ACover{} array.
Our modification is not due to the generalization.
In fact their algorithm works for computing \ACover{s} as it is.
We changed their algorithm just for simplicity.
We will briefly discuss the difference of their and our algorithms later.

Li and Smyth showed some properties of longest $=$-cover arrays, but not all of them hold under SCERs.
For instance, the longest \ACover{} array in Table~\ref{table:example_parameter} is a counterexample to Theorems~2.2 and~2.3 in~\cite{Li2002}.
So it is not trivial that their algorithm and our variant work under SCERs
and we need to carefully check the correctness of the algorithms.

Their algorithm involves an auxiliary array of length $n$ based on the notion of ``live'' prefixes.
A prefix $S$ of $T$ is said to be \emph{live} if $T$ can be extended so that $S$ will be a cover of $TU$ for some $U \in \Sigma^*$.
This notion is also known as ``left seeds''~\cite{Christou2012,Christou2013}.
We generalize the notion for SCERs as follows.
\begin{definition}[\ALeftSeed{}]\label{def:liveprefix}
	For strings $T$ of length $n$ and $S$ of length $m$,
	$S$ is said to be a \emph{\ALeftSeed{}} of $T$ if
	there exist $k$ and $l$ such that $k\le l < m$, $S \in \Cov(T[:n-k])$ and $S[:l] \approx T[n-l+1:]$.
	We denote by $\LSeed(T)$ the set of all \ALeftSeeds{} of $T$.
\end{definition}
We remark that it is not necessarily true that $\LSeed(T) = \{\,S \mid S \in \Cov(TU) $ for some $ U\,\}$ according to the above definition, contrarily to the case of the identity relation.
Consider the order-isomorphism $\OPapprox$ (Definition~\ref{def:oi}) on $\Sigma = \{{\tt a,b,c,d}\}$ with $\mtt{a} \prec \mtt{b} \prec \mtt{c} \prec \mtt{d}$.
Then $S = \mathtt{acb}$ is a left $\OPapprox$-seed of $T = \mathtt{adcbc}$, since $S \OPapprox T[:3]$ and $S[:2] \OPapprox T[4:]$.
However, for no character $U \in \Sigma$, we have $S \OPapprox (TU)[4:6]$, since $U$ needs to be a character bigger than $\mtt{b}$ and smaller than $\mtt{c}$.

Clearly $\Cov(T) \subseteq \LSeed(T)$. Moreover, $S \in \LSeed(T)$ implies $S \in \LSeed(T')$ for any prefix $T'$ of $T$ unless $|S| > |T'|$.
Being a \ALeftSeed{} is a weaker property than being an \ACover{}, but it is easier to handle in an online algorithm, due to the monotonicity that $T[:j] \notin \LSeed(T[:i-1])$ implies $T[:j] \notin \LSeed(T[:i])$ for every $j < i$.
The following series of lemmas investigate the relation among \ALeftSeed{s} and \ACover{s}.

\begin{lemma}\label{lem:CovLSeed}
If\/ $k \le l$, then $\Cov(T[:n-k]) \cap \LSeed(T[n-l+1:]) \subseteq \LSeed(T)$.
\end{lemma}
\begin{proof}
	Suppose $S \in \Cov(T[:n-k]) \cap \LSeed(T[n-l+1:])$.
	By $S \in \LSeed(T[n-l+1:])$, there are $k',l'$ such that $k' \le l' < |S| \le l$, $S \in \Cov(T[n-l+1:n-k'])$ and $S[:l'] \approx T[n-l'+1:]$.
	We have $S \in \Cov(T[:n-k'])$ by $S \in \Cov(T[:n-k]) $, $n-k \ge n-l$, and \Cref{lem:cover_overlap}.
	Hence $S \in \LSeed(T[:i])$ by \Cref{def:liveprefix}.
\qed
\end{proof}

Lemma~\ref{lem:clearlylseed} says somewhat long prefixes are all \ALeftSeed{s}, which we call \emph{primary}.
Lemma~\ref{lem:deadproperty} says shorter \ALeftSeed{s} are \ACover{s} of long \ALeftSeed{s}.
As a corollary, we obtain Lemma~\ref{lem:setdead}, which corresponds to Lemma~2.5 in~\cite{Li2002}.
\begin{lemma}[Primary \ALeftSeed{s}]
	\label{lem:clearlylseed}
	For any $1 \le i \le n$ and\/ $i-\Border_T[i] \le j \le i$,
    we have $T[:j] \in \LSeed(T[:i])$.
\end{lemma}
\newcommand{\PROOFclearlyseed}{
\begin{proof}
    Let $b = \Border_T[i]$, $m = \lfloor (i-j)/(i-b) \rfloor$, $l = i - (m+1)(i-b)$ and $x_k = k(i-b)+1$ for $k \ge 0$.
	It is enough to show that (a) $\{x_0,\dots,x_m\}$ witnesses $T[:j] \in \Cov(T[:x_m+j-1])$,
	 (b) $T[:l] \approx T[i-l+1:i]$, and (c) $i-(x_m+j-1) \le l < j$.
	 The equation (c) can be verified by simple calculation.

	(a) Since $x_{k+1}-x_k = i-b \le j$, it is enough to show $x_k \in \OCC_{T[:j],T[:i]}$ for all $k \le m$.
	Since $T[:b] \approx T[i-b+1:i]$, any ``corresponding'' substrings of $T[1:b]$ and $T[i-b+1:i]$ are $\approx$-equivalent.
	In particular, $T[x_k:x_k+j-1] \approx T[x_k+i-b:x_k+i-b+j-1] = T[x_{k+1}:x_{k+1}+j-1]$ for all $0 \le k < m$.
	That is, $T[:j] \approx T[x_k:x_k+j-1]$ and thus $x_k \in \OCC_{T[:j],T[:i]}$ for all $0 \le k \le m$.

	(b) The same argument for corresponding substrings of $T[1:b]$ and $T[i-b+1:i]$ of length $l$ establishes
	 $T[:l] \approx T[x_m:x_m+l-1] \approx T[x_{m+1}:x_{m+1}+l-1] = T[i-l+1:i]$.
    \qed
\end{proof}
}
\PROOFclearlyseed{}


\begin{lemma}
	\label{lem:deadproperty}
	For any $1 \le i \le n$,
	$T[:j]$ for $1 \le j < i - \Border_T[i]$ is a \ALeftSeed{} of $T[:i]$ iff
    $T[:j]$ is the longest proper \ACover{} of a \ALeftSeed{} of $T[:i]$.
\end{lemma}
\newcommand{\PROOFdeadproperty}{
\begin{proof}
    Let $b = \Border_T[i]$.
    $(\Longrightarrow)$
    Assume that for $1 \le j < i-b$, $T[:j] \in \LSeed(T[:i])$, namely, there exist $k$ and $l$ such that $k \le l < j$, $T[:j] \in \Cov(T[:i-k])$ and $T[:l] \in \Bord(T[:i])$. 
    Since $T[:b]$ is the longest proper \ABorder{} of $T[:i]$,
    $k \le l \le b$ and $j < i-b \le i-k$.
    By \Cref{lem:longestrecursive}, there exists $T[:m] \in \Cov(T[:i-k])$ such that $j = \LCover_T[m]$.
    Moreover, since $j < m \le i-k$ and $k \le l < m$,
    we have $T[:m] \in \LSeed(T[:i])$. 
    Therefore $T[:j]$ is the longest proper \ACover{} of $T[:m]$, which is a \ALeftSeed{} of $T[:i]$.

    $(\Longleftarrow)$
    Assume there is a \ALeftSeed{} prefix $T[:m]$ of $T[:i]$ that is properly covered by $T[:j]$.
    By \Cref{def:liveprefix},
    there exist $k$ and $l$ such that $k \le l < m$, $T[:m] \in \Cov(T[:i-k])$ and $T[:l] \in \Bord(T[:i])$.
    Thus we have $T[:j] \in \Cov(T[:i-k])$ by \Cref{lem:covercover}.
    If $j \ge l$,
    $T[:j] \in \Cov(T[:i-k])$ and $T[:l] \in \Bord(T[:i])$, which implies $T[:j] \in \LSeed(T[:i])$ by \Cref{def:liveprefix}.
    If $j < l <m$,
	$T[:j] \in \Cov(T[:m]) \subseteq \LSeed(T[:m])$ implies $T[:j] \in \LSeed(T[:l])$.
	By Lemma~\ref{lem:CovLSeed}, $T[:j] \in \LSeed(T[:i])$.
\qed\end{proof}
}
\PROOFdeadproperty{}

\begin{lemma}\label{lem:setdead}
	For any $1 \le i \le n$ and  $1 \le j \le i$,
	$T[:j] \in \LSeed(T[:i])$
	 iff there exists $k$ such that
     $i-\Border_T[i] \le k \le i$ and $j=\LCover_T^q[k]$ for some $q \ge 0$.
\end{lemma}
\begin{proof}
By Lemmas~\ref{lem:longestrecursive}, \ref{lem:clearlylseed} and~\ref{lem:deadproperty}.
\qed\end{proof}

Our algorithm involves an auxiliary array based on the following function $\LLive_T$, which is updated by Lemma~\ref{lem:updatellive+}.
The significance of this function is shown as Lemma~\ref{lem:updatelongestcover}.
\begin{definition}[$\LLive_T(i,j)$]\label{def:longestlivecover}
	For a string $T$,
    define
	 \[
	 \LLive_T(i,j)=
     \max (\{\,l \mid T[:l] \in \LSeed(T[:i]) \cap \Cov(T[:j])\,\} \cup \{0\}).
	 \]
\end{definition}
\begin{lemma}\label{lem:updatelongestcover}
	For any $1 \le i \le n$,
	$\LCover_T[i] = \LLive_T(i,\Border_T[i])$.
\end{lemma}
\begin{proof}
It suffices to show
\(
	\Cov(T[:i]) \setminus [T[:i]]_\approx = \LSeed(T[:i]) \cap \Cov(T[:b])
\)
for $b= \Border_T[i]$.
If $C \in \Cov(T[:i])$ with $|C| \neq i$, then obviously $C \in \Bord(T[:i]) \cap \LSeed(T[:i]) $.
By Lemma~\ref{lem:bordertext}, $C \in \Bord(T[:b])$.
Suppose $S \in \LSeed(T[:i]) \cap \Cov(T[:b])$.
There is $k < |S|$ such that $S \in \Cov(T[:i-k])$.
By $k < |S| \le b$ and Lemma~\ref{lem:cover_overlap}, we have $S \in \Cov(T[:i]) $.
\qed\end{proof}
%
%
%


\begin{lemma}\label{lem:updatellive+}
	\(
		\LLive_T(i,j) = \LLive_T(i-1,j)
	\) for $1 \le j \le \Border_T[i]$.
	Moreover, for $j = \Border_T[i]$, if\/ $T[:j] \notin \LSeed(T[:{i-1}])$, then
	\(
		\LLive_T(i,j) = \LLive_T(i-1,\LCover[j])
	\).
\end{lemma}
\begin{proof}
	Let $l=\LLive_T(i-1,j)$ and $l'=\LLive_T(i,j)$.
	Since $j \le \Border_T[i] < i$, we have $l' < i$, which implies $l' \le l$.

	Suppose $l=0$. This implies $l'=0$ and thus $l'=l$ holds.
	Suppose in addition that $j = \Border_T[i]$ and $T[:j] \notin \LSeed(T[:{i-1}])$.
	The fact $l=0$ means $\LSeed(T[:i-1]) \cap \Cov(T[:j]) = \emptyset$, which implies $\LSeed(T[:i-1]) \cap \Cov(T[:LCover_T[j]]) = \emptyset$ by Lemmas~\ref{lem:shortcoverlong} and~\ref{lem:covercover}.
	Therefore, $\LLive_T(i-1,\linebreak[2]\LCover_T[j]) = 0$.
	So the lemma holds.

	Hereafter we assume $l \ge 1$. Let $b_i=\Border_T[i] \ge 1$.
	By $T[:l] \in \LSeed(T[:i-1])$, there exists $k < l$ such that $T[:l] \in \Cov(T[:i-1-k])$.
	On the other hand, by $b_i \le i-1$, $T[:l] \in \LSeed(T[:b_i])=\LSeed(T[i-b_i+1:i])$.
	Since $k < l \le j \le b_i$, by Lemma~\ref{lem:CovLSeed}, $T[:l] \in \Cov(T[:i-1-k]) \cap \LSeed(T[i-b_i+1:i])$ implies $l \in \LSeed(T[:i])$.
	Thus $l'=l$.

	Suppose $j = b_i$ and $T[:j] \notin \LSeed(T[:{i-1}])$.
	Since $\Cov(T[:j]) = \Cov(T[:\LCover_T[j]]) \cup [T[:j]]_{\approx}$ by Lemmas~\ref{lem:shortcoverlong} and~\ref{lem:covercover},
	$T[:j] \notin \LSeed(T[:{i-1}])$ implies $\LLive_T(i-1,\LCover_T[j]) = l = l'$.
\qed\end{proof}

\begin{algorithm2e}[t]
	\caption{Algorithm computing the longest \ACover{} array} 
	\label{alg:lcover}
	\SetVlineSkip{0.5mm}
	let $\Border$ be the \ABorder{} array of $T$\;
	$\LiveChildren[i] \leftarrow 0$,
	$\LargestLive[i] \leftarrow i$ for $0 \le i \le n$\;
	\For{$1 \le i \le n$}{
		\If{$\LiveChildren[\Border[i]]=0$ {\bf and} $0 < 2 \cdot \Border[i] < i$\label{algln:ifdead}}{
			$\LargestLive[\Border[i]] \leftarrow \LargestLive[\LCover[\Border[i]]]$\;\label{alg:longest_LargeLive}
		}
		$\LCover[i] \leftarrow \LargestLive[\Border[i]]$\;\label{algln:lcover}
		$\LiveChildren[\LCover[i]] \leftarrow \LiveChildren[\LCover[i]] + 1$\;\label{alg:livechildren++}
		\If{$i > 1$}{
			$c_1 \leftarrow i-\Border[i]$\;
			$c_2 \leftarrow (i-1) - \Border[i-1]$\;
            \For{$j$ \textbf{\textup{from}} $c_2$ \textbf{\textup{to}} $c_1-1$\label{algln:inner_for}}{
    			\While{$\LiveChildren[j] = 0$\label{alg:dead_condition}}{
					$\LiveChildren[\LCover[j]] \leftarrow \LiveChildren[\LCover[j]] - 1$\;\label{alg:livechildren--}
					$j \leftarrow \LCover[j]$\;
				}
            }
		}
	}
\end{algorithm2e}
%

\Cref{alg:lcover} computes the longest \ACover{} array $\LCover_T$ of $T$ as $\LCover$ taking the \ABorder{} array $\Border_T$ as input.
Following Li and Smyth~\cite{Li2002}, we explain the algorithm using a tree formed by $\LCover_T$, called the \emph{$\approx$-cover tree}.
The $\approx$-cover tree consists of nodes $0,\dots,n$.
The root is $0$ and the parent of $j \neq 0$ is $\LCover_T[j]$.
By Lemma~\ref{lem:longestrecursive}, $T[:k] \in \Cov(T[:j])$ if and only if $k \neq 0$ and $k$ is an ancestor of $j$ (including the case where $k=j$) in the $\approx$-cover tree.
Hereafter, we casually use the index $j$ to mean (any string $\approx$-equivalent to) the prefix $T[:j]$ of $T$, if no confusion arises.
We use two additional arrays $\LiveChildren$ and $\LargestLive$, which have zero-based indices in accordance with the $\approx$-cover tree's nodes.
$\LiveChildren[j]$ counts the number of children of $j$ that are \ALeftSeed{s} of $T$.
$\LargestLive[j]$ points at the lowest ancestor of $j$ that is a \ALeftSeed{} of $T$.
More precisely, the algorithm maintains them so that they satisfy the following invariants at the end of the $i$-th iteration of the outer \textbf{for} loop.
\begin{enumerate}
   \item $\LargestLive[j] = j$ if $\LiveChildren[\Border_T[j]] > 0$ or $\Border_T[j] \ge i-\Border_T[i]$ or $\Border_T[j]=0$ for $0 \le j \le n$.\label{inv:LL1}
   \item $\LargestLive[j] = \LLive_T(i,j)$ for $0 \le j \le \Border_T[i]$.\label{inv:LL2}
   \item $\LCover[j]=\LCover_T[j]$ for $1 \le j \le i$.\label{inv:Cov}
   \item $\LiveChildren[j] = |\LSC(i,j)|$,
    where
\[
	\LSC(i,j) = \{\, k \mid T[:k] \in \LSeed(T[:i]) \text{ and } j = \LCover_T[k] \,\}
\,,\]
	for $0 \le j \le n$.
	Note that $\LSC(i,j) = \emptyset$ for $j \ge i$.\label{inv:Child}
\end{enumerate}
Suppose we already have the $\approx$-cover tree for $T[:i-1]$.
To update it for $T[:i]$ by adding a node $i$, we must determine the parent $\LCover_T[i]$ of $i$.
By Lemma~\ref{lem:updatelongestcover} and the invariant, we know that $\LCover_T[i] = \LargestLive[\Border_T[i]]$.
The array $\LargestLive$ can be maintained by Lemma~\ref{lem:updatellive+}, where we must update $\LargestLive[j]$ when $T[:j] \notin \LSeed(T[:{i-1}])$ for $j = \Border_T[i] > \Border_T[i-1]$.
By Lemma~\ref{lem:setdead}, $T[:j] \in \LSeed(T[:{i-1}])$ iff $i-1 - \Border_T[i-1] \le j \le i-1$ or $\LiveChildren[j] > 0$ assuming that $\LiveChildren$ satisfies the invariant for $i-1$.
Therefore, constructing the $\approx$-cover tree is reduced to maintaining the array $\LiveChildren$.
By Lemma~\ref{lem:setdead}, $\LiveChildren[j]$ counts the number of children of $j$ that are ancestors of an element of the set $P_{i}=\{\, k \mid i-\Border_T[i] \le k \le i\,\}$, which is the index range of primary \ALeftSeed{s}.
At the beginning of the $i$-th iteration, $\LiveChildren$ is based on $P_{i-1}$, and we must update $\LiveChildren$ to be based on $P_i$ by the end of the $i$-th iteration.
 $\LiveChildren[j]$ needs to be updated only when $j$ is an ancestor of some $k$ in the difference of $P_{i-1}$ and $P_i$.
So, we first increment the value $\LiveChildren[\LCover[i]]$ by one as $\LCover[i]$ has got a new child $i \in P_i \setminus P_{i-1}$.
Since $\LCover[i]$ is a \ALeftSeed{} of $T[:i-1]$, we need not increment $\LiveChildren[j]$ for further ancestors $j$ of $\LCover[i]$.
For those $k \in P_{i-1} \setminus P_{i}$, we decrement $\LiveChildren[\LCover[k]]$ unless $k$ is an ancestor of $P_i$.
If this results in $\LiveChildren[\LCover[k]]=0$, we recursively decrement $\LiveChildren[\LCover^2[k]]$, and so on.

\begin{example}
We consider the parameterized-equivalence $\PRapprox$ (Definition~\ref{def:pr}) as an SCER.
Suppose we have computed the $\PRapprox$-cover tree for $T[:5]=\mtt{abcac}$ as shown in Figure~\ref{fig:covertree}~(a).
Our goal is to obtain the one for $T[:6]=\mtt{abcacc}$ shown in Figure~\ref{fig:covertree}~(d).
Since $\LargestLive[j]=j$ for all $j$ throughout this example, we focus on updating $\LCover$ and $\LiveChildren$.
In the figures,  $\LiveChildren$ is shown in parentheses beside each node.
 We have $\Border_T[5]=2$ and $\Border_T[6]=1$,
 so the index sets of the primary left $\PRapprox$-seeds of $T[:5]$ and $T[:6]$ are $P_5=\{3,4,5\}$ and $P_6 = \{5,6\}$, respectively.
 Since $\Border_T[6]=1$, Algorithm~\ref{alg:lcover} first lets $\LCover[6] = \LargestLive[\Border[6]] = 1$.
 In other words, a new node $6$ is added as a child of $1$.
It remains to update $\LiveChildren$, which is now based on $P_5=\{3,4,5\}$ but shall be based on $P_6=\{5,6\}$.
First we increment $\LiveChildren[\LCover[6]]=\LiveChildren[1]$ by one, as illustrated in \Cref{fig:covertree} (b).
 At this moment, $\LiveChildren[j]$ counts the number of children of $j$ which are ancestors of some of $P_5 \cup P_6=\{3,4,5,6\}$.
The inner \textbf{for}-loop of \Cref{algln:inner_for} modifies $\LiveChildren$ so that it shall be based on $\{4,5,6\}$ first and then on $\{5,6\}$.
Since the node $3$ is the parent of $4$, the $\LiveChildren$ arrays based on $\{3,4,5,6\}$ and $\{4,5,6\}$ are identical, as shown in Figures~\ref{fig:covertree}~(b) and~(c), respectively.
To modify $\LiveChildren$ to be based on $\{5,6\}$, we decrement $\LiveChildren[j]$ if $j$ has a child which is an ancestor of $4$ but not that of $5$ or $6$.
Since the node 4 is such a child of $\LCover[4]=3$ (4 is an ancestor of 4, and $\LiveChildren[4]=0$ means that $4$ is not an ancestor of 5 or 6), so we decrement $\LiveChildren[3]$ by one.
This results in $\LiveChildren[3]=0$, by which we know that the node $3$ is an ancestor of 4 but not that of 5 or 6.
Hence we decrement $\LiveChildren[\LCover[3]] = \LiveChildren[2]$.
This results in $\LiveChildren[2]=1$, which means that the node $2$ is an ancestor of 4 and that of 5 or 6 at the same time.
So, we stop the recursion and obtain the $\approx$-cover tree with $\LiveChildren$ based on $P_6=\{5,6\}$, as shown in \Cref{fig:covertree}~(d).
\begin{figure}[t]\centering
	\newcommand{\anch}[2]{(#2)}
	\newcommand{\bl}[1]{\textcolor{red}{\bf{#1}}}
	\tikzstyle{lseed}=[draw=red,ultra thick]
	\begin{tikzpicture}[scale=0.5, every node/.style=circle,thick,fill=white,inner sep=0pt,minimum size=4mm]\small
	\draw(1.5, -1.5) node{(a)};
	\node[draw] (0) at (1.6,5.2) [label={[rectangle]east:\anch{}{1}}] {0};
	\node[draw] (1) at (1.6,3.9) [label={[rectangle]east:\anch{}{1}}] {1};
	\node[draw] (2) at (1.6,2.6) [label={[rectangle]east:\anch{}{2}}] {2};
	\node[lseed] (3) at (0.8,1.3) [label={[rectangle]east:\anch{}{1}}] {\bl{3}};
	\node[lseed] (4) at (0,0) [label={[rectangle]east:\anch{}{0}}] {\bl{4}};
	\node[lseed] (5) at (2.8,1.3) [label={[rectangle]east:\anch{}{0}}] {\bl{5}};
	\draw[-,red,ultra thick] (1) to (0);
	\draw[-,red,ultra thick] (2) to (1);
	\draw[-,red,ultra thick] (3) to (2);
	\draw[-,red,ultra thick] (4) to (3);
	\draw[-,red,ultra thick] (5) to (2);
	\end{tikzpicture}
\quad
	\begin{tikzpicture}[scale=0.5, every node/.style=circle,thick,fill=white,inner sep=0pt,minimum size=4mm]\small
	\draw(1.8, -1.5) node{(b)};
	\node[draw] (0) at (2.1,5.2) [label={[rectangle]east:\anch{}{1}}] {0};
	\node[draw] (1) at (2.1,3.9) [label={[rectangle]east:\anch{}{2}}] {1};
	\node[draw] (2) at (1.4,2.6) [label={[rectangle]east:\anch{}{2}}] {2};
	\node[lseed] (3) at (0.7,1.3) [label={[rectangle]east:\anch{}{1}}] {\bl{3}};
	\node[lseed] (4) at (0,0) [label={[rectangle]east:\anch{}{0}}] {\bl{4}};
	\node[lseed] (5) at (2.8,1.3) [label={[rectangle]east:\anch{}{0}}] {\bl{5}};
	\node[lseed] (6) at (3.5,2.6) [label={[rectangle]east:\anch{}{0}}] {\bl{6}};
	\draw[-,red,ultra thick] (1) to (0);
	\draw[-,red,ultra thick] (2) to (1);
	\draw[-,red,ultra thick] (3) to (2);
	\draw[-,red,ultra thick] (4) to (3);
	\draw[-,red,ultra thick] (5) to (2);
	\draw[-,red,ultra thick] (6) to (1);
	\end{tikzpicture}
\quad
	\begin{tikzpicture}[scale=0.5, every node/.style=circle,thick,fill=white,inner sep=0pt,minimum size=4mm]\small
	\draw(1.8, -1.5) node{(c)};
	\node[draw] (0) at (2.1,5.2) [label={[rectangle]east:\anch{}{1}}] {0};
	\node[draw] (1) at (2.1,3.9) [label={[rectangle]east:\anch{}{2}}] {1};
	\node[draw] (2) at (1.4,2.6) [label={[rectangle]east:\anch{}{2}}] {2};
	\node[draw] (3) at (0.7,1.3) [label={[rectangle]east:\anch{}{1}}] {3};
	\node[lseed] (4) at (0,0) [label={[rectangle]east:\anch{}{0}}] {\bl{4}};
	\node[lseed] (5) at (2.8,1.3) [label={[rectangle]east:\anch{}{0}}] {\bl{5}};
	\node[lseed] (6) at (3.5,2.6) [label={[rectangle]east:\anch{}{0}}] {\bl{6}};
	\draw[-,red,ultra thick] (1) to (0);
	\draw[-,red,ultra thick] (2) to (1);
	\draw[-,red,ultra thick] (3) to (2);
	\draw[-,red,ultra thick] (4) to (3);
	\draw[-,red,ultra thick] (5) to (2);
	\draw[-,red,ultra thick] (6) to (1);
	\end{tikzpicture}
\quad
	\begin{tikzpicture}[scale=0.5, every node/.style=circle,thick,fill=white,inner sep=0pt,minimum size=4mm]\small
	\draw(1.8, -1.5) node{(d)};
	\node[draw] (0) at (2.1,5.2) [label={[rectangle]east:\anch{}{1}}] {0};
	\node[draw] (1) at (2.1,3.9) [label={[rectangle]east:\anch{}{2}}] {1};
	\node[draw] (2) at (1.4,2.6) [label={[rectangle]east:\anch{}{1}}] {2};
	\node[draw] (3) at (0.7,1.3) [label={[rectangle]east:\anch{}{0}}] {3};
	\node[draw] (4) at (0,0) [label={[rectangle]east:\anch{}{0}}] {4};
	\node[lseed] (5) at (2.8,1.3) [label={[rectangle]east:\anch{}{0}}] {\bl{5}};
	\node[lseed] (6) at (3.5,2.6) [label={[rectangle]east:\anch{}{0}}] {\bl{6}};
	\draw[-,red,ultra thick] (1) to (0);
	\draw[-,red,ultra thick] (2) to (1);
	\draw[-] (3) to (2);
	\draw[-] (4) to (3);
	\draw[-,red,ultra thick] (5) to (2);
	\draw[-,red,ultra thick] (6) to (1);
	\end{tikzpicture}
	\caption{
	Updating the $\PRapprox$-cover tree of $T[:5]=\mtt{abcac}$ (a) for that of $T[:6]=\mtt{abcacc}$ (d).
	$\LiveChildren$ counts the numbers of children which are ancestors of some nodes drawn as thick red circles.
	Those highlighted nodes represent primary left $\PRapprox$-seeds $\{3,4,5\}$ of $T[:5]$ in (a) and those $\{5,6\}$ of $T[:6]$ in (d).
	Paths from highlighted nodes to the root are highlighted, so that $\LiveChildren[j]$ is the number of highlighted edges from $j$.
	\label{fig:covertree}}
\end{figure}
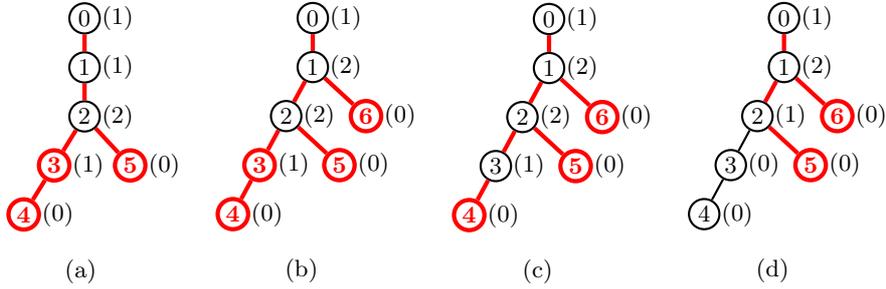
\end{example}
We remark that Li \& Smyth's original algorithm maintains an array $\Dead$ that represents whether $j \notin \LSeed(T[:i])$ in addition to the arrays used in our algorithm.
Our algorithm judges the property using two arrays $\Border$ and $\LiveChildren$ based on Lemmas~\ref{lem:clearlylseed} and~\ref{lem:deadproperty}.
The reason why their algorithm requires the additional array is that it performs the inner \textbf{for} loop of \Cref{algln:inner_for} in the reverse order.
If we perform the loop in the reverse order without the auxiliary array, in the above example,
in the iteration on $j=4$, we obtain the tree in \Cref{fig:covertree} (d),
and then in the iteration on $j=3$, the value of $\LiveChildren[\LCover[3]]=\LiveChildren[2]$ is decremented to $0$ and further more $\LiveChildren[\LCover[2]]=\LiveChildren[1]$ is decremented to 1.
Their algorithm stops iteration of the \textbf{while} loop at Line~\ref{alg:dead_condition} if $\Dead[j]=\True$, to restrain excessive decrement of $\LiveChildren[j]$.

\begin{theorem}\label{thm:lcover}
	Given the \ABorder{} array $\Border_T$ of $T$,
	\Cref{alg:lcover} computes the longest \ACover{} array $\LCover_T$ of $T$ in $O(n)$ time.
\end{theorem}
\begin{proof}
	We prove the above invariants by induction on $i$.
	In the first iteration, neither of the \textbf{if} antecedents are satisfied.
	At the end of the iteration, we have
	 $\LCover[1]=\LargestLive[\Border[1]]=\LargestLive[0]=0$ and $\LiveChildren[0]=1$.
	Together with the initialization, all the arrays satisfy the above invariants.
	By Lemmas~\ref{lem:lcover_invariant1} and~\ref{lem:lcover_invariant2}, finally the algorithm computes $\LCover_T$.
	The linear-time complexity is shown in Lemma~\ref{lem:lcover_time}.
\qed\end{proof}

\begin{corollary}
	If $\Border_T$ can be computed in $\beta(n)$ time,
	$\LCover_T$ can be computed in $O(\beta(n) + n)$ time.
\end{corollary}

\begin{lemma}\label{lem:lcover_invariant1}
Suppose that all the invariants hold at the beginning of the $i$-th iteration of the outer \textbf{for} loop.
Then, at the end of the $i$-th loop, the invariants on $\LargestLive$ and $\LCover$ are satisfied.
\end{lemma}
\begin{proof}
	Assume that
	$\LiveChildren$, $\LargestLive$, and $\LCover$ hold the above properties at the end of the $(i-1)$-th iteration.
    Let $b_{i} = \Border[i]$ and $b_{i-1} = \Border[i-1]$.

    We first show that the invariant on $\LargestLive$ is satisfied.
    Concerning the first claim on $\LargestLive$,
    the value of $\LargestLive[j]$ can be altered from its initial value $j$ only when
    $\LiveChildren[j] = 0$, $0 < 2j < i$ and $j=b_i$,
    in which case, the invariant does not necessitate $\LargestLive[j] = j$.
    On the other hand, by \Cref{lem:barrayproperty},
    if $\Border_T[j] < i-1-\Border_T[i-1]$, then $\Border_T[j] < i-\Border_T[i]$.
    Therefore, once the value of $\LargestLive[j]$ has been altered from $j$, the invariant will never necessitate $\LargestLive[j] = j$.

    Concerning the second claim on $\LargestLive$, suppose $j \le b_i$.
    If $j< b_i$, then $j \le b_{i-1}$ by \Cref{lem:barrayproperty}.
    By the induction hypothesis on $\LargestLive[j]$ and Lemma~\ref{lem:updatellive+}, $\LargestLive[j]=\LLive(i-1,j)=\LLive(i,j)$.
    It remains to show $\LargestLive[b_i]=\LLive(i,b_i)$.

	If $b_i = 0$, $\LLive_T(i,b_{i}) = \LLive_T(i-1,b_{i}) = 0$.
	Suppose $b_i > 0$ and $T[:b_{i}] \notin \LSeed(T[:i-1])$.
	Let $m=\LCover[b_i]$, for which $m < b_i \le b_{i-1}+1$.
    By \Cref{lem:updatellive+} and the induction hypothesis on $\LargestLive[m]$, we have $\LLive_T(i, b_i) = \LLive_T(i-1, m) = \LargestLive[m]$.
	By Lemmas~\ref{lem:clearlylseed} and~\ref{lem:deadproperty} and the induction hypothesis, $b_i < i-1-b_{i-1}$ and $\LiveChildren[b_i] = 0$.
	Thus, since $2b_i \le b_i + b_{i-1}+1 < i-1+1 = i$,
	the algorithm lets $\LargestLive[b_{i}] = \LargestLive[m]$ in Line~\ref{alg:longest_LargeLive}, which fulfills the invariant on $\LargestLive$.

    Suppose $T[:b_{i}] \in \LSeed(T[:i-1])$.
    In this case, there is $k < b_i$ such that $T[:b_i] \in \Cov(T[:i-1-k])$.
    By Lemma~\ref{lem:cover_overlap}, $T[:b_{i}] \in \Cov(T[:i]) \subseteq \LSeed(T[:i])$ and thus $\LLive_T(i,b_i) = b_i$.
    By Lemma~\ref{lem:barrayproperty}, $b_i = b_{i-1}+1$ holds.
    By Lemmas~\ref{lem:clearlylseed} and~\ref{lem:deadproperty}, either $b_i \ge i-1-b_{i-1}$ or $\LiveChildren[b_i] > 0$.
	The former case implies $2 b_i \ge i$ and thus in either case the algorithm does not execute Line~\ref{alg:longest_LargeLive}.
	By the induction hypothesis, $\LargestLive[b_i]=b_i$, which fulfills the invariant on $\LargestLive[b_i]$.

	The invariant on $\LCover$ is fulfilled in Line~\ref{algln:lcover}, which makes $\LCover[i] = \LargestLive[b_{i}]$ in accordance with \Cref{lem:updatelongestcover}.
\qed\end{proof}

\begin{lemma}\label{lem:lcover_invariant2}
If the invariants hold at the beginning of the $i$-th iteration of the outer \textbf{for} loop,
 the invariant on $\LiveChildren$ holds at the end of the $i$-th loop.
\end{lemma}
\begin{proof}
	Assume that at the end of the $(i-1)$-th iteration, the invariants hold.
    Let $b_{i} = \Border[i]$, $b_{i-1} = \Border[i-1]$, $c_1 = i-b_{i}$, and $c_2 = (i-1)-b_{i-1}$.
    Note that $c_1 \ge c_2$ by Lemma~\ref{lem:barrayproperty}.

    First we discuss $\LiveChildren[j]$ for $j \ge c_1$.
	For any $k$ with $c_2 \le c_1 \le k < i$, by Lemma~\ref{lem:clearlylseed}, $T[:k] \in \LSeed(T[:i-1]) \cap \LSeed(T[:i])$.
	This means that for any $j$ with $c_2 \le c_1 \le j \le i$,
    \[
    \LSC(i,j) = \LSC(i-1,j) \cup I_j
    \]
    where $I_j=\{i\}$ for $j= \LCover_T[i]$ and $I_j=\emptyset$ for $j \neq \LCover_T[i]$.
    Accordingly, for those $j \ge c_1$, the algorithm realizes $\LiveChildren[j] =|\LSC(i-1,j)|+|I_j| = |\LSC(i,j)|$.

	It remains to show the invariants on $\LiveChildren[j]$ for $j < c_1$.
	By \Cref{lem:setdead}, $\LSC(i,j)$ can be rewritten as $\LSC(i,j)=\lsc(c_1,i,j)$ for
\[
		\lsc(k,l,j) = \LCover^{-1}[j] \cap \{\,\LCover^q[h] \mid k \le h \le l \text{ and }q \ge 0\,\}
\]
	where $\LCover^{-1}[j] = \{\, h \mid j = \LCover[h] \,\}$.
	In terms of  the \ACover{} tree, $\LCover^{-1}[j]$ is the set of children of $j$
	and $\lsc(k,l,j) $ is the set of children which have an element between $k$ and $l$ as a descendant (a node is thought to be a descendant of itself).
	Note that $0 \notin \LCover^{-1}[j]$ for any $j \ge 0$.
	After executing Line~\ref{alg:livechildren++} of Algorithm~\ref{alg:lcover}, together with the induction hypothesis, we have
	$\LiveChildren[j] = | \lsc(c_2,i,j) |$.
	If $c_1=c_2$, then the algorithm does not go into the inner \textbf{for} loop of Line~\ref{algln:inner_for}
	and we have done the proof.
	If $c_1>c_2$,
	it is enough to show that at the end of each iteration of the inner \textbf{for} loop of Line~\ref{algln:inner_for},
	\begin{equation}\label{eq:lc_sd}
		\LiveChildren[l] = | \lsc(j+1,i,l) |
    \end{equation}
	for all $l < c_1$.
	For $j=c_1-1$, we have $\LiveChildren[l] = |\lsc(c_1,i,l)| = |\LSC(i,l)|$ for all $l < c_1$.
	For this purpose, we show by induction on $r$ that
	 at the end of the $r$-th iteration of the \textbf{while} loop (\Cref{alg:dead_condition}), we have
    \begin{align}
    \text{\scalebox{0.98}[1]{$
		\LiveChildren[l] = \big|
	\lsc(j+1,i,l) \cup (\LCover^{-1}[l] \cap \{\LCover^q[j] \mid q \ge r \})\big|
	$}}
\label{eq:sd_sd}
	\end{align}
	for all $l < c_1$.
	Note that there always exists $r_j$ such that $\LCover^{r_j}[j]=0$, for which
	 $\LCover^{-1}[l] \cap \{\LCover^q[j] \mid q \ge r_j \} = \emptyset$, i.e., Eq.~(\ref{eq:sd_sd}) is equivalent to (\ref{eq:lc_sd}).

	For $r=0$, i.e., at the beginning of the first iteration of the \textbf{while} loop,
	Eq.~(\ref{eq:lc_sd}) for $j-1$ holds, i.e., $\LiveChildren[l] = | \lsc(j,i,l) |$, which is equivalent to (\ref{eq:sd_sd}) with $r=0$.

	Assuming the induction hypothesis (\ref{eq:sd_sd}) for $r$ holds, we show that it is the case for $r+1$.
	Increasing $r$ by one never expands the set on the right hand of (\ref{eq:sd_sd}).
	The set will lose an element $h$ iff $h = \LCover^{r}[j]$, $l=\LCover^{r+1}[j]$ and
	\begin{equation}
		\LCover^{r}[j] \notin \{\LCover^q[k] \mid j < k \le i ,\ q \ge 0\}
	\,.\label{eq:decrement}
	\end{equation}

	If $\LiveChildren[\LCover^{r}[j]] \neq 0$,
	the loop is not repeated.
	It is enough to show that for any $l < c_1$
	\begin{equation}\label{eq:lc_sd-}
		\LCover^{-1}[l] \cap \{\LCover^q[j] \mid q \ge r \} \subseteq \lsc(j+1,i,l)
	\,,\end{equation}
	 so that we establish (\ref{eq:lc_sd}).
	If $\LCover^{r}[j]=0$, $\LCover^{-1}[l] \cap \{\LCover^q[j] \mid q \ge r \} = \emptyset$.
	Clearly (\ref{eq:lc_sd-}) holds.
	Suppose $\LCover^{r}[j]\neq 0$.
	The assumption that $\LiveChildren[\LCover^{r}[j]] \neq 0$ means, by induction hypothesis (\ref{eq:sd_sd}), there is
	\begin{align*}
	k \in {}& \lsc(j+1,i,\LCover^r[j])
	\\ & \cup (\LCover^{-1}[\LCover^r[j]]\cap\{\LCover^{q}[j] \mid q \ge r\})
	\,.\end{align*}
	By $\LCover^{-1}[\LCover^r[j]]\cap\{\LCover^{q}[j] \mid q \ge r\}=\emptyset$,
	$k \in \lsc(j+1,i,\LCover^r[j])$, which means
	 $k = \LCover^s[h] \in \LCover^{-1}[\LCover^r[j]]$ for some $j< h \le i$ and $s \ge 0$,
	  i.e., $\LCover^{s+1}[h] = \LCover^{r}[j]$.
	For $1 \le l \le c_1$,
	if $\LCover^{q}[j] \in \LCover^{-1}[l]$ for some $q \ge r$, then
	\[
		 \LCover^{q-r+s+1}[h] = \LCover^{q}[j] \in \LCover^{-1}[l]\,.
	\]
	That is, $ \LCover^{q}[j] \in \lsc(j+1,i,l)$, which shows (\ref{eq:lc_sd-}) and thus (\ref{eq:lc_sd}).

	Suppose $\LiveChildren[\LCover^{r}[j]] = 0$.
	We show that (\ref{eq:decrement}) holds.
	By the induction hypothesis (\ref{eq:sd_sd}) for $r$, $\LiveChildren[\LCover^{r}[j]] = 0$ means
	\begin{align*}
		& \lsc(j+1,i,\LCover^{r}[j])
\\		& \cup (\LCover^{-1}[\LCover^{r}[j]] \cap \{\LCover^q[j] \mid q \ge r \})
		 = \emptyset
		 \,.
	\end{align*}
	If (\ref{eq:decrement}) did not hold, there were $j'$ and $q$ such that $\LCover^{r}[j]=\LCover^{q}[j']$ and $j < j' \le i$,
	where $q \ge 1$ by $\LCover^{r}[j] \le j < j'$.
	Then $\LCover^{q-1}[j'] \in \LCover^{-1}[\LCover^{r}[j]]$, which is a contradiction.
	So, the condition~(\ref{eq:decrement}) holds.
\qed\end{proof}

\begin{lemma}\label{lem:lcover_time}
\Cref{alg:lcover} runs in $O(n)$ time.
\end{lemma}
\begin{proof}
	Let $t(j)$ and $f(j)$ be the numbers of times that the \textbf{while} condition on $j$ (\Cref{alg:dead_condition}) is judged true and false, respectively.
	Since $\sum_{j=0}^n f(j) \le n + \sum_{j=1}^n t(j)$, it is enough to show $t(j) \le 1$ for every $j$ to establish the linear-time complexity.
	Suppose that the algorithm finds $\LiveChildren[j]=0$ at the \textbf{while} loop in the $i$-th iteration of the outer \textbf{for} loop.
	We show that it happens for the least $i > j$ such that $T[:j] \notin \LSeed(T[:i])$.
	Note that the condition is checked only for $j < c_1$, where $c_1 = i - \Border[i]$.
	Therefore, $\LiveChildren[j]=0$ implies $T[:j] \notin \LSeed(T[:i])$ by Lemma~\ref{lem:deadproperty}.
	Since $T[:j] \notin \LSeed(T[:i])$ implies $T[:j] \notin \LSeed(T[:i'])$ for any $i' > i$,
	it is enough to show $T[:j] \in \LSeed(T[:i-1])$.
	For $c_2 = i-1-\Border[i-1]$,
	by the algorithm, $j = \LCover^q[k]$ for some $c_2 \le k < c_1$ and $q \ge 0$.
	If $q=0$, i.e., $c_2 \le j=k <c_1$, by Lemma~\ref{lem:clearlylseed}, $T[:j] \in \LSeed(T[:i-1])$.
	If $q \ge 1$, the value  $\LiveChildren[j]$ is decremented in the $q$-th iteration of the \textbf{while} loop, just before deciding  $\LiveChildren[j]=0$.
	Moreover, $T[:j] \notin \LSeed(T[:i])$ implies $j \neq \LCover[i]$, and hence $\LiveChildren[j]$ was strictly positive at the end of the $(i-1)$-th iteration of the outer \textbf{for} loop.
	By the invariant, $\LSC(i-1,j) \neq \emptyset$, which means $T[:j] \in \LSeed(T[:i-1])$.
\qed\end{proof}

\bibliographystyle{splncs04}
\bibliography{ref}

\end{document}